%% file: main.tex
\documentclass[10pt,conference]{IEEEtran}
\usepackage[cmex10]{amsmath}
\usepackage{amsfonts}
\usepackage{soul,color}

\input{c0ommands}

\usepackage{wrapfig}
\usepackage{graphicx}
\definecolor{bisque}{rgb}{1.0, 0.89, 0.77}
\definecolor{buff}{rgb}{0.94, 0.86, 0.51}
\sethlcolor{buff}
\begin{document}

\title{Equivalence Checking and Simulation By Computing Range Reduction}

\author{\IEEEauthorblockN{Eugene Goldberg} 
\IEEEauthorblockA{
eu.goldberg@gmail.com}}
\maketitle

\input{abstract}

\input{introduction}
\input{c1rr}

\input{c2rr_by_pqe}

\input{ec_by_crr}

\input{t1wo_ways}

\input{s1im_by_excl}

\input{background}

\input{c3onclusions}

\section*{Acknowledgment}
This research was supported in part by  NSF grants
CCF-1117184 and CCF-1319580.
\bibliographystyle{plain}
\bibliography{short_sat,local}
\end{document}

%% file: c0ommands.tex
\newtheorem{definition}{Definition} 
\newtheorem{proposition}{Proposition}

\newcommand{\pnt}[1]{{\mbox{\boldmath $#1$}}}
\newcommand{\cof}[2]{\mbox{$#1_{\boldsymbol{#2}}$}}

\newcommand{\s}[1]{\mbox{$\{#1\}$}}

\newcommand{\nGz}[2]{$G_{non-\{z\}}$}

\newcommand{\mi}[1]{\mathit{#1}}
\newcommand{\ti}[1]{\textit{#1}}
\newcommand{\tb}[1]{\textbf{#1}}

\newcommand{\ttt}{\>\>\>}

\newcommand{\Tt}{\>\>}

\newcommand{\Sup}[2]{\mbox{$#1^\mi{#2}$}}
\newcommand{\prob}[2]{\mbox{$\exists{#1} [#2]$}}
\newcommand{\Prob}[3]{\mbox{$\exists{#1}\exists{#2}[#3]$}}

\newcommand{\Comment}[1]{}

\newcommand{\Impl}[2]{\mbox{$#1 \rightarrow #2$}}

\newcommand{\Rng}[2]{\mbox{$\mi{Rng}(#1,#2)$}}

\newcommand{\EQ}[2]{\mbox{$\mi{EQ}(#1,#2)$}}
\newcommand{\eco}{EC\_CDCL}
\newcommand{\ecn}{EC\_CRR}
\newcommand{\SE}{\ti{SimByExcl}}

%% file: abstract.tex
\begin{abstract}
We introduce new methods of equivalence checking and simulation based
on Computing Range Reduction (CRR). Given a combinational circuit $N$,
the CRR problem is to compute the set of outputs that disappear from
the range of $N$ if a set of inputs of $N$ is excluded from
consideration. Importantly, in many cases, range reduction can be
efficiently found even if computing the entire range of $N$ is
infeasible.

Solving equivalence checking by CRR facilitates generation of proofs
of equivalence that mimic a ``cut propagation'' approach. A limited
version of such an approach has been successfully used by commercial
tools. Functional verification of a circuit $N$ by simulation can be
viewed as a way to reduce the complexity of computing the range of
$N$.  Instead of finding the entire range of $N$ and checking if it
contains a bad output, such a range is computed only for one
input. Simulation by CRR offers an alternative way of coping with the
complexity of range computation. The idea is to exclude a subset of
inputs of $N$ and compute the range \textit{reduction} caused by such an
exclusion. If the set of disappeared outputs contains a bad one, then
$N$ is buggy.
\end{abstract}

%% file: introduction.tex
\section{Introduction}
 The objective of this paper is to emphasize the importance of developing
 efficient algorithms for Computing Range Reduction (CRR). Earlier we
 showed that CRR can be used for model
 checking~\cite{tech_rep_crr}. Adding equivalence checking and
 simulation to the list of problems that can be handled by CRR makes
 the need for developing efficient CRR algorithms more obvious.

\subsection{The CRR problem}
Let $N(X,Y,Z)$ be a combinational circuit where $X,Y,Z$ are the sets
of input, intermediate and output variables of $N$ respectively (see
Figure~\ref{fig:one_circ}).  We will refer to the set of all outputs
produced by $N$ as the range of $N$.  By an ``output'' here we mean a
complete assignment of values to the variables of $Z$. We will also
use term ``input'' to denote a complete assignment to the variables of
$X$.

\input{one_circ.fig}

Suppose that one excludes a set $K$ of complete assignments to $X$
from the set of available inputs to $N$. If an output \pnt{z} is
produced only by inputs that are in $K$, excluding the inputs of $K$
from consideration leads to disappearance of \pnt{z} from the range of
$N$.  \tb{The CRR problem} is to compute the outputs of $N$ that are
excluded from the range of $N$ if the inputs of a set $K$ are excluded
from consideration.

\subsection{Equivalence checking by CRR}
\label{subsec:ec_by_crr}
Equivalence Checking can be a hard problem even for combinational
circuits. This is especially true when the circuits to be compared
have few functionally equivalent internal points.  Our motivation here
is that, as we argue in Section~\ref{sec:ec_by_crr}, equivalence
checking by CRR facilitates construction of ``natural'' proofs. Such
proofs are generated by industrial equivalence checkers when the
circuits to be compared have a lot of internal points that are
functionally equivalent.

\input{e1c_two_indep.fig}

Application of CRR to equivalence checking is based on the following observation. Let
$N'$ and $N''$ be single-output combinational circuits to be checked
for equivalence. (The expression ``a single-output circuit'' means
that this circuit has only one output variable. Whether we use term
``output'' to refer to an ``output variable'' or ``output assignment''
should be clear from the context.)  Let $M$ be a two-output circuit
composed of $N'$ and $N''$ as shown in Figure~\ref{fig:two_indep}.
Note that the sets $X'$ and $X''$ of input variables of $N'$ and $N''$
are independent of each other.  Then, if $N'$ and $N''$ are not
constants, the range of circuit $M$ consists of all four assignments
to the output variables $z'$,$z''$ of $N'$ and $N''$.

Let (\pnt{x'},\pnt{x''}) denote an assignment to variables $X' \cup
X''$. Suppose that one excludes from consideration all assignments
(\pnt{x'},\pnt{x''}) where $\pnt{x'} \neq \pnt{x''}$. Then if $N'$ and
$N''$ are functionally equivalent, such constraint on inputs of $M$
should lead to disappearing assignments $(z'=0,z''=1)$ and
$(z'=1,z''=0)$ from the range of $M$. If this is not the case, then
there is an input (\pnt{x'},\pnt{x''}) of $M$ where \pnt{x'}=\pnt{x''}
for which $N'$ and $N''$ evaluate to different values i.e.  $N'$ and
$N''$ are inequivalent.

\subsection{Simulation by CRR}
\label{subsec:sim_by_crr}
Functional verification of a combinational circuit $N$ comes down to
checking if the range of $N$ contains an erroneous output.  The
straightforward approach of computing the range of $N$ is extremely
inefficient.  Simulation can be viewed as a way to simplify the range
computation problem by reducing the set of inputs of $N$ to only one
(regular simulation) or to a subset of inputs (symbolic
simulation~\cite{SymbolSim}).  In this paper, we consider a different
way to cope with complexity of range computation called
simulation-by-exclusion.  It is based on the idea of computing the
\ti{reduction of range} of $N$ caused by excluding some inputs from
consideration rather than the entire range of $N$.

Suppose that one needs to check if an erroneous output \pnt{z} is
produced by $N$.  Assume that a set of inputs $K$ is excluded from
consideration. Let $Q$ be the set of outputs that disappear from the
range of $N$ due to exclusion of inputs from $K$.  One can have only
the following two cases. The first case is $\pnt{z} \not\in Q$.  This
means that either \pnt{z} is not in the range of $N$ i.e. $N$ is
correct or $N$ is buggy but there is an input $\pnt{x} \not\in K$
producing \pnt{z}. So exclusion of the inputs of $K$ does not make
\pnt{z} disappear from the range of $N$. Then one can safely remove
the inputs of $K$ from further consideration because if $N$ is buggy,
the set of available inputs still contains a counterexample. The
second case is $\pnt{z} \in Q$.  Then $N$ is buggy because \pnt{z} is
in its range.

In this paper, we describe the basic algorithm of
simulation-by-exclusion and a few of its modifications.

\subsection{CRR by partial quantifier elimination}
In~\cite{tech_rep_crr}, we showed that the CRR problem comes down to
Partial Quantifier Elimination (PQE). The difference between partial
and complete elimination is that, in PQE, only a part of the formula
is taken out of the scope of quantifiers.  Let us explain the relation
between CRR and PQE in the context of equivalence checking. Let
$\mi{EQ}(X',X'')$ denote a propositional formula such that
$\mi{EQ}(\pnt{x'},\pnt{x''}) = 1$ iff \pnt{x'} = \pnt{x''}. Equivalence checking by CRR comes down
to taking $\mi{EQ}$ out of the scope of quantifiers in formula
\prob{W}{\mi{EQ} \wedge F' \wedge F''}. Here formulas $F'$ and $F''$
specify circuits $N'$ and $N''$ of Figure~\ref{fig:two_indep} and $W =
X' \cup X'' \cup Y' \cup Y''$.

Taking $\mi{EQ}$ out of the scope of quantifiers means finding a
quantifier-free formula $H(z',z'')$ such that $\prob{W}{\mi{EQ} \wedge
  F' \wedge F''} \equiv H \wedge \prob{W}{F' \wedge F''}$. If an
output of circuit $M$ of Figure~\ref{fig:two_indep} disappears from
its range after excluding inputs falsifying $\mi{EQ}$, this output
falsifies $H$. So if $H(0,1)=H(1,0)=0$, circuits $N'$ and $N''$ are
equivalent. 

Let us explain the benefits of solving PQE. We assume here and
henceforth that all formulas are represented in the Conjunctive Normal
Form (CNF).  The appeal of PQE is that it allows one to focus on
deriving implications that matter.  Namely, formula $H$ can be built
solely from clauses that are implied by formula $\mi{EQ} \wedge F'
\wedge F''$ and are not implied by formula $F' \wedge F''$. This means
that a PQE solver can ignore the resolvents obtained only from clauses
of $F' \wedge F''$ focusing on resolvents that are descendants of
clauses of $\mi{EQ}$.  Such zooming in on the ``right'' resolvents can
make building formula $H$ much more efficient.

%
%
\subsection{Difference in CRR for equivalence checking  and simulation}
Our methods of using CRR for equivalence checking and simulation are
different in one important aspect. In equivalence checking, one builds
a circuit $M$ of Figure~\ref{fig:two_indep} and eliminates the
\ti{spurious} behaviors where $\pnt{x'}\neq \pnt{x''}$. The presence
of a bug corresponds to the situation when an erroneous output \ti{is
  not excluded} from the range of $M$ under the input constraints.
Conversely, simulation-by-exclusion is based on elimination of
\ti{valid} inputs.  Then, the presence of a bug corresponds to the
situation where an erroneous output \ti{is excluded} from the range of
the circuit under test.

The difference above is important for the following reason. When
describing application of CRR to equivalence checking and simulation
in Subsections ~\ref{subsec:ec_by_crr} and~\ref{subsec:sim_by_crr} we
assumed that range reduction is computed precisely. However, a PQE
solver computes a superset of the set of excluded outputs
~\cite{tech_rep_crr} that may include outputs that are not in the
range (see Subsection~\ref{subsec:pqe_qual}.)  We will refer to such
outputs as noise.  Adding noise to the set of excluded outputs in
equivalence checking by CRR cannot trigger false alarms. (As we
mentioned above the situation of concern here is when a bad output of
$M$ is not excluded under input constraints. ) Conversely, in
simulation-by-exclusion, adding noise can set off a false alarm. This
happens when a bad output \pnt{z} is in the set of excluded outputs
computed by a PQE solver but, in reality, \pnt{z} is not in the range
of the circuit.

%
%
\subsection{Structure of the paper}
This paper is structured as follows. Section~\ref{sec:crr_def}
formally defines the CRR problem. Solving this problem by PQE is
discussed in Section~\ref{sec:crr_by_pqe}. In
Section~\ref{sec:ec_by_crr}, we introduce a method of equivalence
checking by CRR.  Simulation-by-exclusion is discussed in
Sections~\ref{sec:rng_cmp_compl} and ~\ref{sec:sim_by_excl}. Some
background and conclusions are given in Section~\ref{sec:background}
and Section~\ref{sec:conclusions} respectively.

%% file: one_circ.fig.tex
\setlength{\intextsep}{4pt}
\begin{wrapfigure}{l}{1.0in}
 \begin{center}
    \includegraphics[width=0.9in]{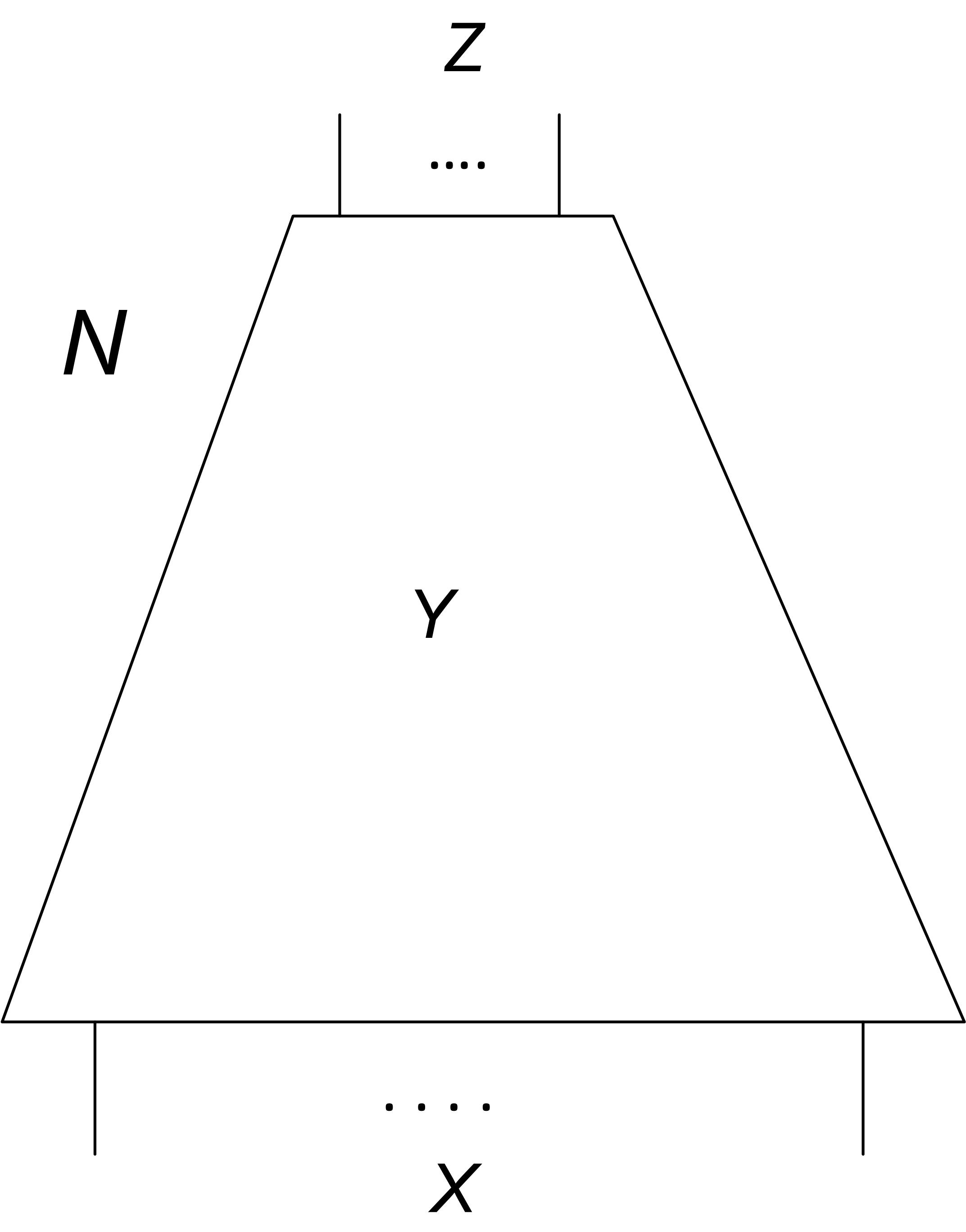}
  \end{center}
\vspace{-4pt}
\caption{Combinational circuit $N$}
\label{fig:one_circ}
\end{wrapfigure}

%% file: e1c_two_indep.fig.tex
\setlength{\intextsep}{4pt}
\begin{wrapfigure}{l}{1.5in}
 \begin{center}
    \includegraphics[width=1.4in]{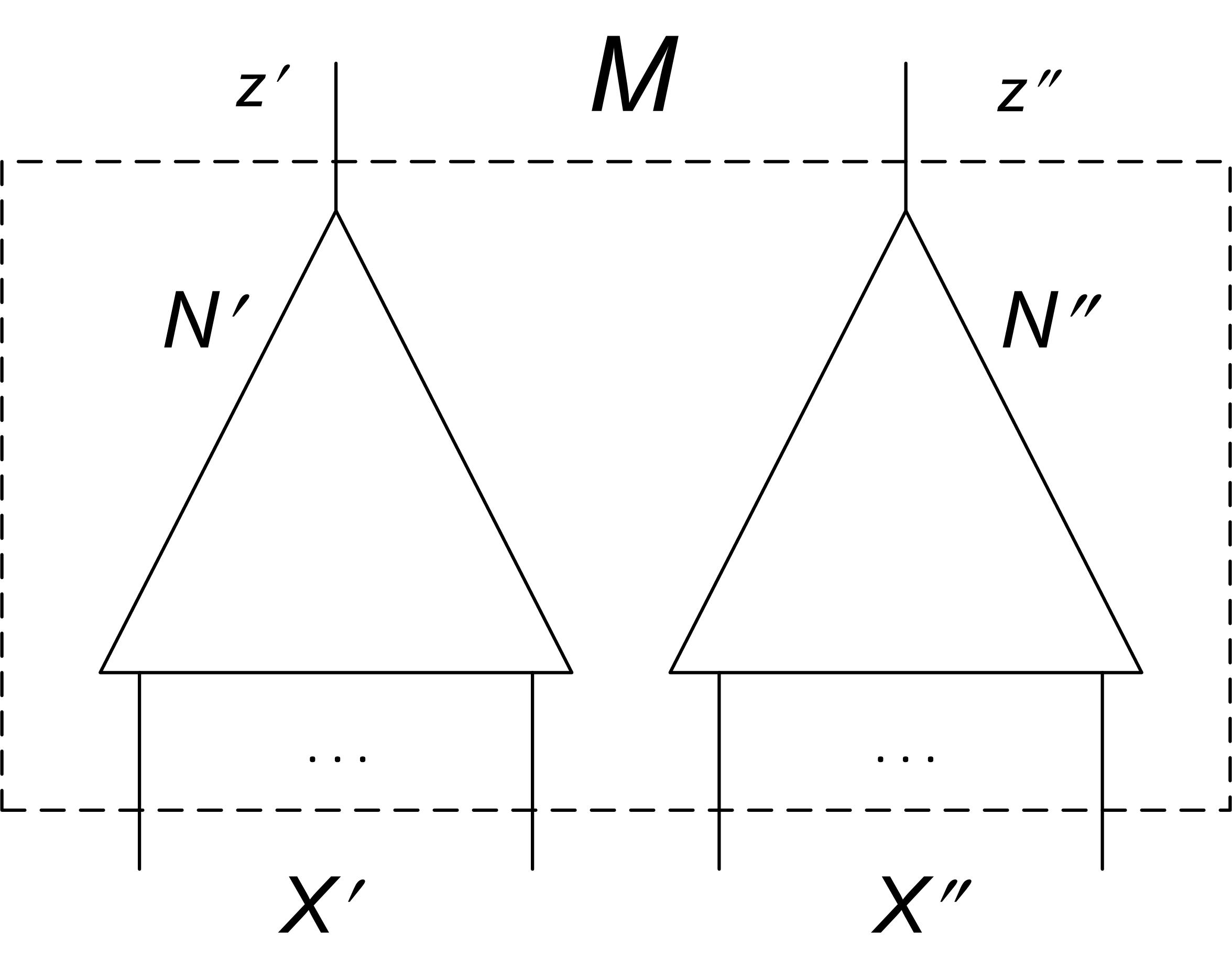}
  \end{center}
\caption{$M$ is a two-output circuit composed of independent circuits $N'$ and $N''$}
\label{fig:two_indep}
\end{wrapfigure}

%% file: c1rr.tex
\section{Problem Of Computing Range Reduction}
\label{sec:crr_def}
Let $N(X,Y,Z)$ be a combinational circuit where $X,Y,Z$ are input,
intermediate and output variables respectively. Let $N$ evaluate to
output \pnt{z} (i.e. a complete assignment to $Z$) under an input
\pnt{x} (i.e. a complete assignment to $X$).  We will say that \pnt{z}
\ti{is produced} by $N$.

Let $A(X)$ be a set of inputs of $N$. Denote by {\boldmath
  $\mi{Rng}(N,A)$} \tb{the range} of $N$ under inputs of $A$.  That is
an output $\pnt{z}$ is in \Rng{N}{A} iff \pnt{z} is produced by $N$
for some input $\pnt{x} \in A$.  If $A$ consists of all $2^{|X|}$
inputs of $N$, then \Rng{N}{A} specifies the entire range of $N$.
From the definition of \Rng{N}{A} it follows that $\Rng{N}{A \setminus
  B} \subseteq \Rng{N}{A}$ where $B$ is a set of inputs of $N$.
\begin{definition}
\label{def:crr_prob}
Let $A$ and $B$ be sets of inputs of circuit $N$ such that $A \cap B
\neq \emptyset$. The \tb{CRR problem} is to find the set $\Rng{N}{A }
\setminus \Rng{N}{A \setminus B}$.  If an output \pnt{z} is in
$\Rng{N}{A } \setminus \Rng{N}{A \setminus B}$, we will say that
\pnt{z} is \ti{excluded from} \Rng{N}{A} \ti{due to excluding inputs}
of $B$.
\end{definition}

Below, we give definition of an approximate solution to the CRR
problem.  The reason is that, as we showed in~\cite{tech_rep_crr},
algorithms for partial quantifier elimination provide only approximate
solutions to the CRR problem (see Subsection~\ref{subsec:pqe_qual}).
\begin{definition}
\label{def:approx_crr}
Let $Q$ denote $\Rng{N}{A} \setminus \Rng{N}{A \setminus B}$ from
Definition~\ref{def:crr_prob}.  Let $Q^*$ be a superset of $Q$.  We
will say that $Q^*$ is an \tb{approximate solution} of the CRR problem
if for every $\pnt{z} \in Q^* \setminus Q$ it is true that $\pnt{z}
\not\in \Rng{N}{A}$.  In other words, an approximate solution of the
CRR problem is the precise solution modulo outputs that are not in
\Rng{N}{A}. We will refer to such outputs as \tb{noise}.
\end{definition}

%% file: c2rr_by_pqe.tex
\section{Computing Range Reduction By Partial Quantifier Elimination}
\label{sec:crr_by_pqe}
In this section, we relate range computation with quantifier
elimination and range reduction computation with partial quantifier
elimination. Besides, we recall our previous results on complete and
partial quantifier elimination.
%
%
\subsection{Complete and partial quantifier elimination}
\label{subsec:pqe_def}
Given a quantified CNF formula \prob{W}{F(V,W)}, the problem of
\tb{Quantifier Elimination} (\tb{QE}) is to find a quantifier-free
formula $H(V)$ such that $H \equiv \prob{W}{F}$.

Given a quantified CNF formula \prob{W}{G(V,W) \wedge F(V,W)}, the
problem of \tb{Partial Quantifier Elimination} (\tb{PQE}) is to find a
quantifier-free CNF formula $G^*(V)$ such that $G^* \wedge \prob{W}{F}
\equiv \prob{W}{G \wedge F}$.  We will say that formula $G^*$ is
obtained by \tb{taking} \pnt{G} \tb{out of the scope of quantifiers}
in \prob{W}{G \wedge F}. The term ``partial'' emphasizes the fact that,
in contrast to QE, only a part of the quantified formula is taken out
of the scope of quantifiers.

In this paper, we do not discuss algorithms for solving the PQE
problem.  In Subsection~\ref{subsec:method_comp}, we only explain how
a PQE algorithm implements a cut advancement strategy. More
information can be found in ~\cite{hvc-14}.

%
%
\subsection{Computing range by QE and range reduction by PQE}
\label{subsec:crr_by_pqe} 
Let $F(X,Y,Z)$ be a CNF formula specifying circuit $N(X,Y,Z)$. That is
every consistent assignment of Boolean values to variables of $N$ is a
satisfying assignment for $F$ and vice versa. Let $H(Z)$ be a CNF
formula such that $H \equiv \Prob{X}{Y}{F}$. Then the outputs of $N$
satisfying $H$ specify the range of $N$. So computing the range of $N$
comes down to solving the QE problem.  Assume that only inputs
satisfying formula $G(X)$ are considered. Then a CNF formula $H(Z)$
equivalent to \Prob{X}{Y}{G \wedge F} specifies the range of $N$ under
inputs satisfying $G$.

Suppose that we are interested only in finding how the range of $N$
\ti{reduces} if the inputs falsifying $G$ are excluded i.e.  only
inputs satisfying $G$ are considered.  As we showed in
~\cite{tech_rep_crr}, this problem comes down to finding CNF formula
$G^*(Z)$ such that $G^* \wedge \prob{W}{F} \equiv \prob{W}{G \wedge
  F}$ where $W = X \cup Y$.  Here $G^*$ specifies an approximate
solution to the CRR problem (see Definition~\ref{def:approx_crr}) i.e.
\begin{itemize}
\item the outputs falsifying $G^*$ specify a superset of the set of
  outputs disappearing from the range of $N$ due to excluding inputs
  falsifying $G$
\item every output \pnt{z} falsifying $G^*$ either disappeared from the
  range of $N$ due to excluding inputs falsifying $G$ or is not in the
  range of $N$
\end{itemize}

%
%
\subsection{Complexity of computing range and range reduction}
\label{subsec:complexity}
As we mentioned above, computing the range of $N$ reduces to QE
whereas finding range reduction comes down to PQE. So the complexity
of computing range and range reduction is directly related to that
of QE and PQE.

The difference between QE and PQE can be expressed in terms of
computing clause redundancy~\cite{fmcad13,hvc-14}.  For the sake of
simplicity, in descriptions of algorithms of ~\cite{fmcad13} and
~\cite{hvc-14} given below we omitted the fact that they use
branching. (First the problem is solved in subspaces. Then the results
of branches are merged to produce a solution that holds for the entire
search space.)

As we showed in~\cite{fmcad12,fmcad13}, to eliminate quantifiers from
\prob{W}{F(X,Y,Z)} where $W = X \cup Y$ one needs to find a CNF
formula $H(Z)$ implied by $F$ that makes all the clauses of $F$ with
quantified variables redundant in \prob{W}{F \wedge H}. Then formula
$H$ is logically equivalent to \prob{W}{F}.  We will refer to clauses
containing variables of $W$ (i.e. quantified variables) as
\pnt{W}\tb{-clauses}.

Clauses of $H$ are built from $F$ by adding resolvent clauses.  If a
resolvent is a $W$-clause, its redundancy has to be proved along with
the original $W$-clauses of $F$. Eventually, a sufficient number of
clauses depending only on free variables (i.e. those of $Z$) is added
to make the new and old $W$-clauses of $F$ redundant.

Now, let us consider the PQE problem of taking formula $G(X)$ out of
the scope of quantifiers in \prob{W}{G \wedge F}. As we showed
in~\cite{hvc-14}, this problem comes down to finding formula $G^*(Z)$
implied by $G \wedge F$ that makes redundant the clauses of $G$ in
\prob{W}{G^* \wedge G \wedge F}. The clauses of $G^*$ are built by
resolving clauses of $G \wedge F$.  If a resolvent is a $W$-clause and
is a descendant of a clause of $G$, its redundancy needs to be proved
as well. On the other hand, $W$-resolvents produced solely from
clauses of $F$ do not need to be proved redundant. This can make PQE
\ti{drastically simpler} than QE if $G$ is much smaller than $F$.
%
%
\subsection{Quality of PQE-solving}
\label{subsec:pqe_qual}
In this subsection, we clarify the relation between the quality of a
PQE-algorithm and that of an approximate solution to the CRR
problem. As in the previous subsection, we assume that formula
$G^*(Z)$ is obtained by taking formula $G(X)$ out of the scope of
quantifiers in \prob{W}{G \wedge F}.  Here $W = X \cup Y$ and $F$
specifies circuit $N(X,Y,Z)$.

The definitions below are based on the following observation. If $G^*$
is a solution to the PQE problem above, a formula obtained by adding
to $G^*$ a clause $C(Z)$ implied by $F$ is also a solution to the PQE
problem.  One can view $C$ as ``noise''. Such a clause is falsified
only by outputs that are not in the range of $N$. In general, the set
of outputs falsifying a clause of $G^*$ can contain outputs that are
in the range of $N$ and those that are not.

%
%

\begin{definition}
Let $C$ be a clause of $G^*$. We will say that $C$ is \tb{noise-free}
if every output falsifying $C$ is in the range of $N$. Otherwise,
clause $C$ is called \tb{noisy}. Formula $G^*$ is called a noisy
solution to the PQE problem if it contains a noisy clause. Otherwise,
$G^*$ is called noise-free.
\end{definition}

%% file: ec_by_crr.tex
\section{Equivalence 
Checking By Computing Range Reduction}
\label{sec:ec_by_crr}
In this section, we describe equivalence checking by CRR. In
Subsection~\ref{subsec:main_prop}, we give the main proposition on
which our method is based.  Subsection~\ref{subsec:method_comp}
compares equivalence checking by CRR and SAT-solving.

%
%
\subsection{Main proposition}
\label{subsec:main_prop}
The intuition behind applying CRR to equivalence checking was described in the
introduction.  In this subsection, we substantiate this intuition in a
formal proposition.

Let $N'(X',Y',z')$ and $N''(X'',Y'',z'')$ be single-output
combinational circuits to be checked for equivalence.  Let
$M(X',X'',Y',Y'',z',z'')$ be the circuit composed of $N'$ and $N''$
as shown in Figure~\ref{fig:two_indep}.  Let $F'(X',Y',z')$ and
$F''(X'',Y'',z'')$ be CNF formulas specifying $N'$ and $N''$
respectively. Assume that $X' = \s{x'_1,\dots,x'_k}$ and $X'' =
\s{x''_1,\dots,x''_k}$. Denote by \EQ{X'}{X''} formula
$(x'_1 \equiv x''_1) \wedge \dots \wedge (x'_k \equiv x''_k)$ that is
satisfied by inputs \pnt{x'} and \pnt{x''} to $N'$ and $N''$ iff
\pnt{x'} = \pnt{x''}.

%
%
\begin{proposition}
\label{prop:ec_pqe}
Let $H(z',z'')$ be a CNF formula obtained by taking \EQ{X'}{X''} out
of the scope of quantifiers from the formula \prob{W}{\mi{EQ} \wedge
  F' \wedge F'' } where $W = X' \cup X'' \cup Y' \cup Y''$.  Circuits
$N'$ and $N''$ are functionally equivalent iff one of the two
conditions below hold:
\begin{enumerate}
\item $H(0,1) = H(1,0) = 0$ 
\item Circuits $N'$ and $N''$ are identical constants (i.e. $N' \equiv
  N'' \equiv 0$ or $N' \equiv N'' \equiv 1$.).
\end{enumerate}
\end{proposition}
\begin{proof}
\noindent\tb{The if part}.  Assume that $H(0,1)=H(1,0)=0$.  Since
$H(z',z'')$ is a solution to the PQE problem, it is implied by formula
$\mi{EQ} \wedge F' \wedge F''$. This means that for any consistent
assignments $(\pnt{x'},\pnt{y'},z')$ and $(\pnt{x''},\pnt{y''},z'')$
to $N'$ and $N''$ such that \pnt{x'}=\pnt{x''} the values of $z'$ and
$z''$ are equal to each other.  Thus $N'$ and $N''$ are functionally
equivalent.  If the second condition holds i.e.  $N'$ and $N''$ are
identical constants, then $N',N''$ are obviously equivalent.

\noindent\tb{The only if part}. Assume that $N'$ and $N''$ are
functionally equivalent.  We need to prove that one of the conditions
above holds. Assume the contrary i.e. that $H(0,1)=H(1,0)=0$ does not
hold and $N'$, $N''$ are not identical constants.

Let us assume, for the sake of clarity, that $H(0,1)=1$. (The case
when $H(1,0)=1$ can be considered in a similar manner.) This means
that formula $H$ does not contain a clause $C(z',z'')$ falsified by an
assignment $z'=0,z''=1$.  Let us consider the following two reasons
for that.

First, $C$ is not implied by $\mi{EQ} \wedge F' \wedge F''$. This means
that there is an assignment \pnt{x'}=\pnt{x''} to $N'$ and $N''$ for
which $N'$ evaluates to 0 and $N''$ evaluates to 1. Hence $N'$ and $N''$
are inequivalent and we have a contradiction.

Second, $C$ is implied by both $\mi{EQ} \wedge F' \wedge F''$ and $F'
\wedge F''$ and one does not need to add clause $C$ to make formula
$\mi{EQ}$ redundant. Since $F'$ and $F''$ do not share variables, the
fact that $F' \wedge F''$ implies $C$ means that the circuit $M$ does
not produce output $z'=0,z''= 1$ even when inputs of $M$ are not
constrained by $EQ$. This is possible only if at least one of the
circuits $N'$,$N''$ is a constant. Let us consider the following three
cases.
\begin{itemize}
\item $N'$ and $N''$ are identical constants. Contradiction.
\item $N'$ is a constant 1 and $N''$ is either a constant 0 or a non-constant
circuit. Hence $N'\not\equiv N''$ . Contradiction.
\item $N''$ is a constant 0 and $N'$ is either a constant 1 or a non-constant
circuit. So $N' \not\equiv N''$. Contradiction.
\end{itemize}
\end{proof}
%
%
%

\subsection{Equivalence checking by regular SAT-solving and CRR}
\label{subsec:method_comp}
In this subsection, we compare equivalence checking by a SAT-solver
with Conflict Driven Clause Learning (CDCL) and by CRR.

\input{e4c_cut2.fig}

The equivalence checking of circuits $N'(X',Y',z')$ and
$N''(X'',Y'',z'')$ can be performed by checking the satisfiability of
formula $F$ equal to $F'(X,Y',z') \wedge F''(X,Y'',z'') \wedge (z'
\neq z'')$. Indeed, the existence of an assignment
$\pnt{x},\pnt{y'},\pnt{y''},z',z''$ satisfying $F$ means that $N'$ and
$N''$ produce different output values for the same input \pnt{x}. We
will refer to the method of equivalence checking by a SAT solver based
on CDCL solver as EC\_CDCL. Here EC stands for Equivalence Checking.

Proposition~\ref{prop:ec_pqe} justifies the following method of
equivalence checking by CRR.  We will refer to this method as EC\_CRR.  This
method solves the PQE problem specified by formula $F$ equal to
\prob{W}{\EQ{X'}{X''} \wedge F'(X',Y',z') \wedge F''(X'',Y'',z'')}
where $W = X' \cup X'' \cup Y' \cup Y''$.  Namely, EC\_CRR computes
formula $H(z',z'')$ obtained by taking $\mi{EQ}$ out of the scope of
quantifiers in \prob{W}{EQ \wedge F' \wedge F''}.  If $H(0,1) = H(1,0)
= 0$, then $N'$ and $N''$ are equivalent.  Otherwise, one needs to
check if $N'$ and $N''$ are identical constants. If so, then $N'$ and
$N''$ are equivalent, otherwise they are not. Theoretically, proving
$N'$ and $N''$ inequivalent does not require generation of a
counterexample.  The very fact that $H(0,1) = 1$ or $H(1,0)=1$ and
$N',N''$ are not constants guarantees that $N'$ and $N''$ are not
equivalent. However, in a practical implementation of EC\_CRR, a
counterexample may be generated before computation of formula $H$ is
completed.

To distinguish the formulas $F$ solved by EC\_CDCL and EC\_CRR, we
will refer to them as \Sup{F}{cdcl} and \Sup{F}{crr} respectively. The
SAT problem can be viewed as a special case of QE where all variables
of the formula are existentially quantified.  Since \ecn~ is actually
a PQE algorithm, one can compare \ecn~ and \eco~ in terms of computing
clause redundancy as we did Subsection~\ref{subsec:complexity}.  In
those terms, the objective of deriving an empty clause by a SAT-solver
is to make redundant all clauses with quantified variables (i.e. all
clauses containing literals). So the difference between \eco~ and
\ecn~ is that the former is aimed at making redundant \ti{all} clauses
of \Sup{F}{cdcl} while the goal of the latter is to make redundant
only a small subset of clauses of \Sup{F}{crr} (namely those of
$\mi{EQ}$).

The fact that \ecn~ needs to prove redundancy of a smaller set of
clauses than \eco~ does not necessarily imply greater efficiency of
\ecn.  What matters, though, is that solving the PQE problem specified
by \Sup{F}{crr} facilitates generation of proofs based on the cut
advancement strategy shown in Figure~\ref{fig:ec_cut2}. Such proofs are
natural for structurally similar circuits.  Here is a naive version
of how this works. To make clauses of $\mi{EQ}$ redundant one just
needs to produce new clauses obtained by resolving $\mi{EQ}$ with
those of $F' \wedge F''$. Note that a PQE solver has to make redundant
a new resolvent (depending on a quantified variable) \ti{only} if it
is a descendant of a clause from $\mi{EQ}$.  In general, such a clause
will contain variables of both $F'$ and $F''$ relating variables of
$N'$ and $N''$.  The set $G_1$ of new resolvents that made the clauses
of $\mi{QE}$ redundant specify a new ``cut'' that consists of the
variables present in $G_1$. After that a set $G_2$ of descendants of
$G_1$ that make the clauses of $G_1$ redundant is built. This goes on
until the cut consisting of variables $z'$ and $z''$ is reached.

A cut advancement strategy has been successfully used by industrial
equivalence checkers. However, it is applied only when $N'$ and $N''$
are so structurally similar that one can build a cut consisting of
internal points of $N'$ and $N''$ that are functionally equivalent. So
this version of cut advancement is very limited. On the other hand,
EC\_CRR can arguably extend the cut advancement strategy to a much
more general class of equivalence checking problems.

%% file: e4c_cut2.fig.tex
\setlength{\intextsep}{4pt}
\begin{wrapfigure}{l}{1.5in}
 \begin{center}
    \includegraphics[width=1.5in]{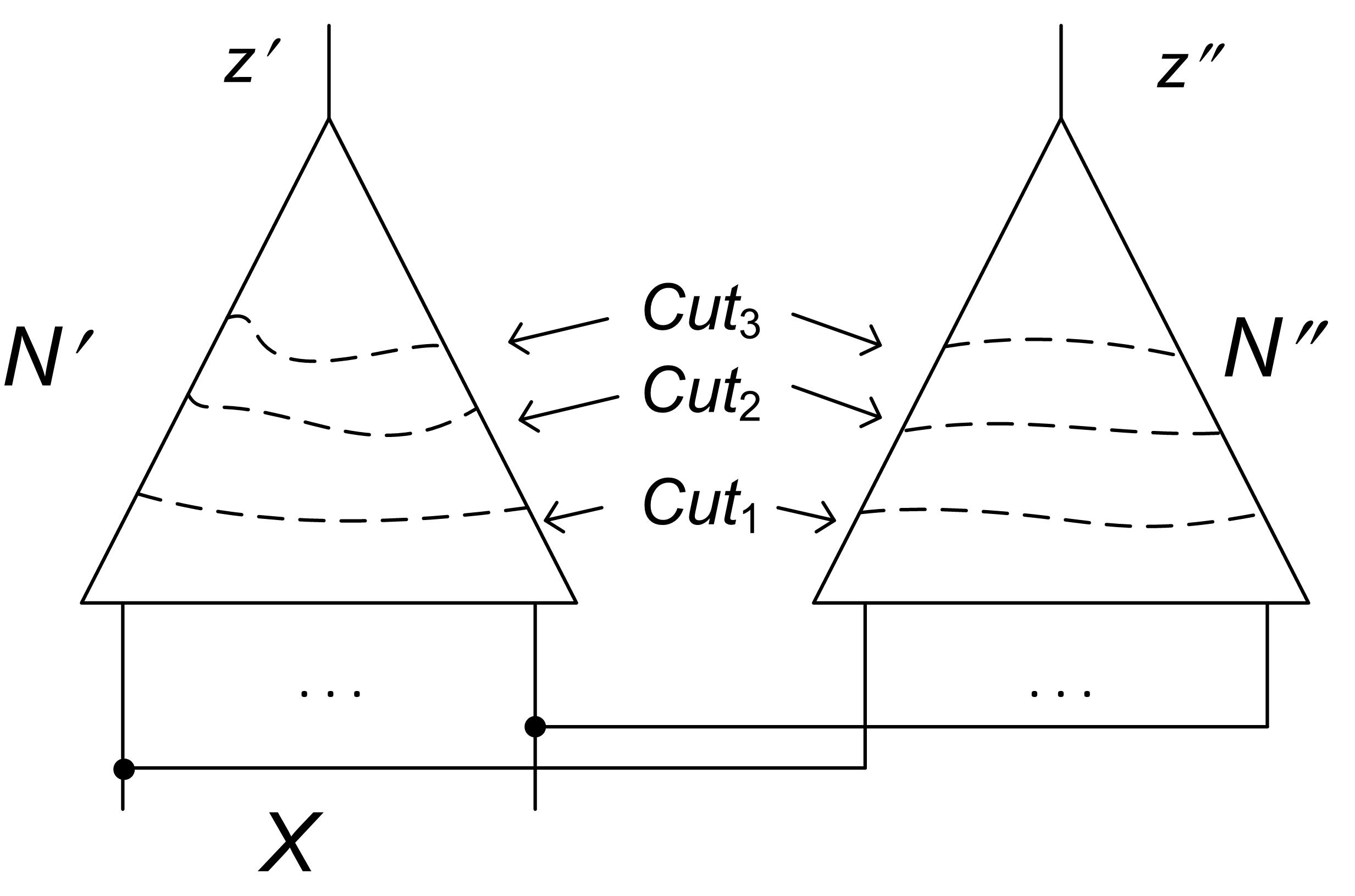}
  \end{center}
\caption{Equivalence checking by cut advancement}
\label{fig:ec_cut2}
\end{wrapfigure}

%% file: t1wo_ways.tex
\section{Two Ways To Handle Complexity Of Range Computation}
\label{sec:rng_cmp_compl}
Function verification of a combinational circuit $N(X,Y,Z)$ comes down
to checking if there is an input for which $N$ produces an erroneous
output. We assume that the erroneous outputs are specified by the user
as a set $E$.  A straightforward way to verify $N$ is to find the
range of $N$ and check if it overlaps with $E$. As we mentioned in
Subsection~\ref{subsec:crr_by_pqe}, finding the range of $N$ comes
down to solving the QE problem. Although, in general, the latter is
very hard, in some cases, it can be efficiently solved for real-life
circuits. For example, if $N$ is a single-output circuit, computing
its range comes down to Circuit-SAT that, in many practical cases, can
be solved by state-of-the-art SAT-solvers.  In
Sections~\ref{sec:rng_cmp_compl} and~\ref{sec:sim_by_excl}, we assume
that functional verification of $N$ is too hard for a SAT-solver.

In this section, we consider two methods for reducing the complexity
of computing range of $N$. One method is to compute range of $N$ for a
subset of inputs. This method is used in traditional simulation. We
introduce another method based on the following idea: compute range
\ti{reduction} caused by excluding some inputs of $N$ rather than the
\ti{entire range}.

%
%
\subsection[title]{Simulation as range computation\footnote{\input{f1ootnote}}}
\label{subsec:sim_as_rng_comp}
Let $F(X,Y,Z)$ be a CNF formula specifying circuit $N$. Computing the
range of $N$ comes down to finding formula $H(Z)$ logically equivalent
to \prob{W}{F} where $W = X \cup Y$. Simulation can be viewed as a way
to decrease the complexity of finding $H(Z)$ by shrinking the set of
allowed inputs to only one. Let $S(X)$ be a CNF formula satisfied by
only one input \pnt{x}.  Then computing the output produced by $N$
comes down to solving the QE problem \prob{W}{S \wedge F}.  Formula
$S$ can be represented as a set of $|X|$ unit clauses satisfied only
by \pnt{x}. (A clause is called unit if it contains exactly one
literal.)  A solution $H(Z)$ for this QE problem can be represented as
a set of $|Z|$ unit clauses satisfied only by the output of $N$ for
input \pnt{x}.

An obvious flaw of traditional simulation is that it provides
information about the value of $N$ only for one input out of
$2^{|X|}$.  One can address this issue by using formulas $S$ with many
satisfying assignments, Then solving the QE problem specified by
\prob{W}{S \wedge F} finds the range of $N$ where all inputs
satisfying $S$ are processed together. Unfortunately, the complexity
of QE grows very fast as the number of inputs satisfying $S$
increases. One way to reduce the complexity of QE when $S$ allows many
inputs is employed in symbolic simulation~\cite{SymbolSim}. The idea
is to pick formulas $S$ for which all inputs satisfying $S$ have the
same execution trace.

%
%
\subsection{Simulation-by-exclusion}
\label{subsec:sim_by_excl_and_crr}
We described the main idea of simulation-by-exclusion in
Subsection~\ref{subsec:sim_by_crr}. Given constraints excluding some
inputs, compute only reduction of range of $N$ caused by this
exclusion rather than the entire range of $N$ under allowed inputs.
Here we continue explanation describing how simulation-by-exclusion is
implemented by PQE. Let the complete assignments to $X$
\ti{falsifying} CNF formula $Q(X)$ specify the inputs excluded from
consideration.  Let CNF formula $Q^*(Z)$ be obtained by taking $Q(X)$
out of the scope of quantifiers in \prob{W}{Q \wedge F} where $W = X
\cup Y$. Then the outputs falsifying $Q^*$ specify the range reduction
of $N$ due to exclusion of inputs falsifying $Q$.

Suppose that no output falsifying $Q^*$ is in the set $E$ of erroneous
outputs of $N$. This implies that either $N$ is correct or there is an
input \pnt{x} producing an output from $E$ (a counterexample) that is
not excluded i.e. \pnt{x} satisfies $Q$. This means that removing all
inputs falsifying $Q$ from consideration is safe because it cannot
remove \ti{all} counterexamples (if any).

As we mentioned earlier, a PQE solver provides only an approximate
solution to the CRR problem. This means that even if an output \pnt{z}
falsifying $Q^*$ is in $E$, one still needs to check if \pnt{z} is in
the range of $N$. A trivial way to do it is to check if formula
$\overline{\cof{C}{z}} \wedge F$ is satisfiable where \cof{C}{z} is
the longest clause falsified by \pnt{z}. If so, clause \cof{C}{z} is
not implied by $F$ and \pnt{z} is in the range of $N$.

Checking the satisfiability of $\overline{\cof{C}{z}} \wedge F$ may be
expensive. However, in the algorithm of simulation-by-exclusion we
introduce in Section~\ref{sec:sim_by_excl}, formula $F$ is constantly
updated by adding clauses excluding inputs of $N$.  This may
drastically simplify the SAT-check above. Besides, there are at least
two techniques to make this SAT-check even simpler. The first
technique is based on the fact that all inputs producing \pnt{z} (if
any) falsify $Q$. So checking if \pnt{z} is in the range of $N$
reduces to testing the satisfiability $\overline{Q} \wedge
\overline{\cof{C}{z}} \wedge F$.

The second technique is based on the following observation. Let $R$ be
a resolution derivation of a clause $C$ of $Q^*$ falsified by
\pnt{z}. Then if \pnt{z} is in the range of $N$, every cut
$\mi{Cut}(R)$ of proof $R$ has to include at least one clause $A$ that
is implied by $Q \wedge F$ but not $F$. (We assume here that $R$ is a
directed graph nodes of which correspond to original clauses of $Q
\wedge F$ or resolvents.) If every clause of $\mi{Cut}(R)$ is implied
by $F$ then the clause $C$ derived by $R$ is implied by $F$ too and so
\pnt{z} is not in the range of $N$.  So to find out if \pnt{z} is in
the range of $N$ it is sufficient to check if formula $\overline{A}
\wedge \overline{Q} \wedge \overline{\cof{C}{z}} \wedge F$ is
satisfiable for every clause $A$ of $\mi{Cut}(R)$. The longer clause
$A$ is the simpler this SAT-check. If all such formulas are
unsatisfiable for $\mi{Cut}(R)$ then \pnt{z} is not in the range of
$N$. Otherwise, an input of $N$ producing \pnt{z} can be extracted
from an assignment satisfying formula $\overline{A} \wedge
\overline{Q} \wedge \overline{\cof{C}{z}} \wedge F$.

%
%
\subsection{Simulation-by-exclusion versus regular simulation}
\label{subsec:comp_two_sims}
In this subsection, term ``regular simulation'' refers to finding a
solution $H(Z)$ to the QE problem specified by \prob{W}{Q \wedge F}
that we introduced in Subsection~\ref{subsec:sim_as_rng_comp}
(see also Footnote~\ref{fnote:one}).

The main difference between regular simulation and
simulation-by-exclusion is as follows. In regular simulation, to get a
result one needs to perform a computation that goes all the way from
inputs to outputs. If only one input \pnt{x} is allowed by $Q$, this
computation produces an execution trace consisting of values assigned
to variables of $N$ when applying input \pnt{x}. On the other hand,
simulation-by-exclusion can get a valuable result even by a local
computation that does not reach the outputs of $N$.

\input{s2im_by_excl.fig}

Let us consider the example of Figure~\ref{fig:sim_by_excl} showing
a circuit $N$. Suppose that lines $x_1,x_2$ feed only gates $G_1$ and
$G_2$ and lines $y_1,y_2$ feed only gate $G_3$.  Note that the output
of $G_3$ specified by $y_3$ implements function $x_1 \equiv x_2$
i.e. $y_3 = 1$ iff values of $x_1$ and $x_2$ are equal.  Let $C$
denote clause $x_1 \vee x_2$. It is not hard to show that $C$ is
redundant in formula \prob{W}{C \wedge F} i.e.  $\prob{W}{C \wedge F}
\equiv \prob{W}{F}$. The redundancy of $C$ means that one can exclude
the inputs falsifying $C$ (i.e.  all the inputs for which $x_1 = 0,
x_2 = 0$) from consideration because such exclusion does not change
the range of $N$. Indeed, since $x_1,x_2,y_1,y_2$ feed only gates
$G_1,G_2,G_3$, circuit $N$ produces the same output for two inputs
that are different only in the values of $x_1$ and/or $x_2$ if these
inputs produce the same assignment to $y_3$.  Since, even after
excluding assignment $x_1=0,x_2=0$, both values of $y_3$ can still be
produced, the set of outputs circuit $N$ can generate remains intact.

Importantly, the fact that clause $C$ can be safely used to exclude
inputs of $N$ is derived \ti{locally} without any knowledge of the
rest of the circuit $N$. Such a result cannot be reproduced
efficiently by regular simulation. To exclude the inputs falsified by
$C$ from consideration one would need to solve the QE problem
specified by \prob{W}{\overline{C} \wedge F}. The complexity of such
QE can be very high if the size of circuit $N$ is large. Of course, if
regular simulation succeeds in performing QE, it might find a
bug. However, if no bug is found, the result is the same as for
simulation-by-exclusion: the inputs falsifying $C$ can be
excluded. However the price of this result in regular simulation can
be dramatically higher than in simulation-by-exclusion.

%% file: f1ootnote.tex
\label{fnote:one}
Traditionally, given a combinational circuit $N(X,Y,Z)$, simulation is
a deterministic procedure for computing the output \pnt{z} produced by
$N$ when an input \pnt{x} is applied. Unfortunately, this procedure
does not have a trivial extension to the case when two or more inputs
are applied \ti{at once}.  For that reason, we use a \ti{semantic}
notion of simulation as a computation of the range of $N$ for a subset
of inputs, omitting a detailed description of this computation.

%% file: s2im_by_excl.fig.tex
\setlength{\intextsep}{4pt}
\begin{wrapfigure}{l}{1.3in}
 \begin{center}
    \includegraphics[width=1.3in]{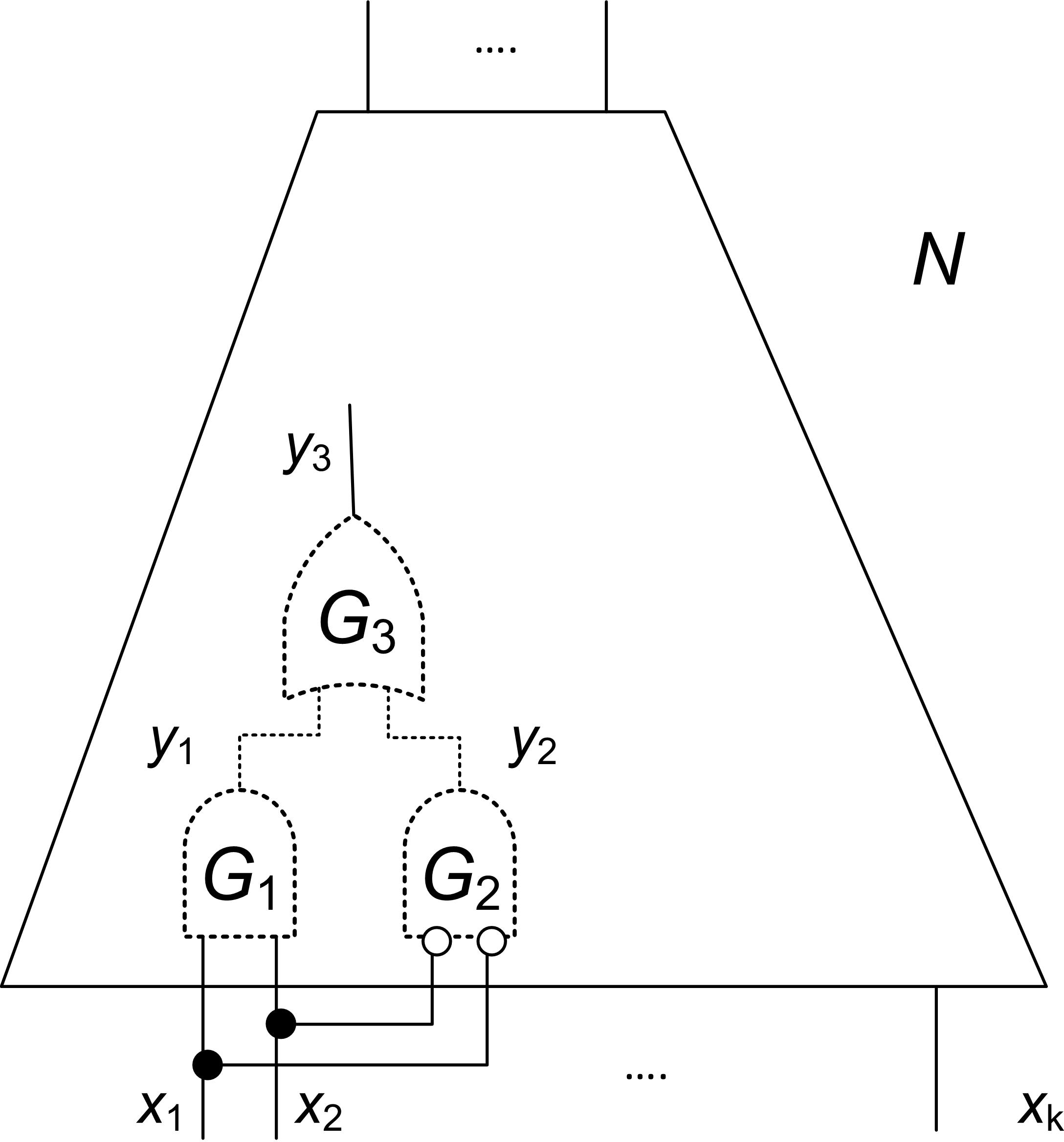}
  \end{center}
\caption{An example illustrating simulation-by-exclusion}
\label{fig:sim_by_excl}
\end{wrapfigure}

%% file: s1im_by_excl.tex
\section{A Few Algorithms Implementing  Simulation-By-Exclusion}
\label{sec:sim_by_excl}
In this section, we describe a few algorithms implementing
simulation-by-exclusion.  The basic algorithm is shown in
Figure~\ref{fig:impl_sim_by_excl} and described in
Subsections~\ref{subsec:alg_descr} and~\ref{subsec:completeness}. Two
modifications of this algorithm are given in
Subsection~\ref{subsec:modifications}. For the sake of clarity, we
will assume that one needs to verify a \ti{single-output} circuit $N$.
We will assume that if $N$ is correct it always evaluates to 0.  An
input for which $N$ evaluates to 1 is a counterexample showing that
$N$ is buggy. The extension of algorithms of this section to
multi-output circuits is straightforward.
%
%
\subsection{Algorithm description}
\label{subsec:alg_descr}
The algorithm of simulation-by-exclusion called \ti{SimByExl} consists
of three parts separated in Figure~\ref{fig:impl_sim_by_excl} by the
dotted lines. In the first part (lines 1-6), an input \pnt{\hat{x}} is
generated for which $N$ evaluates to 0 (lines 2-4).  This input is
never excluded and serves two purposes.  First, it guarantees that $N$
is not a constant 1. Second, keeping \pnt{\hat{x}} allowed prevents
\SE~ from excluding all inputs for which $N$ evaluates to 0. \SE~also
builds formula $F$ specifying $N$ (line 1) and generates
\cof{C}{\hat{x}}, the longest clause falsified by \pnt{\hat{x}} (line
5). Finally, \SE~initializes formula $G$ that accumulates the clauses
excluding inputs of $N$ (line 6).

The second part of \SE~ (lines 7-14) consists of a while loop.  In
every iteration of this loop, a new clause $C(X)$ is generated (line
10).  Clause $C$ excludes at least one new input \pnt{x} of $N$
constructed as an assignment satisfying formula $G \wedge
\cof{C}{\hat{x}}$ (line 8).  If \SE~fails to find such \pnt{x}, all
inputs of $N$ but \pnt{\hat{x}} have been excluded. Then $N$ is
correct since, for the excluded inputs, it takes the same value as for
\pnt{\hat{x}} i.e. 0.

\input{s3at_fig}

After generating clause $C$, \SE~computes CNF formula $C^*(z)$ (line
11).  It is obtained by taking $C$ out of the scope of quantifiers in
\prob{W}{C \wedge F} where $W = X \cup Y$.  Formula $C^*$ can only
consist of clause $\overline{z}$ or be empty (no clauses) in which
case $C^*$ is always true.  Note that formula $C^*$ cannot consist of
unit clause $z$ because \SE~never excludes all inputs for which $N$
evaluates to 0. Hence clause $z$ is not implied by $C \wedge F$. If
$C^*$ is empty, then the inputs falsifying $C$ can be safely excluded
from consideration. So clause $C$ is added to $F$ (line 13) and to $G$
(line 14).

If $C^*$ is not empty and hence equal to $\overline{z}$, \SE~breaks
the loop (line 12) and executes the third part of the algorithm (lines
15-17). In this part, \SE~checks whether the derived clause
$\overline{z}$ is pure noise or $C$ excludes a counterexample. If no
counterexample is excluded by $C$, then clause $\overline{z}$ is just
noise, i.e. \Impl{F}{\overline{z}} and so $N$ is correct. Otherwise,
\SE~extracts a counterexample from an assignment satisfying
$\overline{C} \wedge F \wedge z$.

Note that one can just check whether $N$ is constant 0 by running the
SAT-check on $F \wedge z$ where $F$ is the original formula specifying
$N$. SAT-checking of $F \wedge z \wedge \overline{C}$ can be much
simpler for the following two reasons. First, $F$ is constantly
updated by \SE~ by adding clauses excluding inputs. As we explain in
Subsection~\ref{subsec:modifications} such clauses are obtained by
non-resolution derivation and so are very valuable for a
resolution-based SAT-solver. Second, the negation of clause $C$
consists of unit clauses that can additionally simplify the
SAT-check. One more way to simplify this SAT-check that we described
in Subsection~\ref{subsec:sim_by_excl_and_crr} is to exploit the
resolution derivation of clause $\overline{z}$. For the sake of
simplicity this possibility is not mentioned in
Figure~\ref{fig:impl_sim_by_excl}.

%
%
\subsection{Completeness and soundness of \SE}
\label{subsec:completeness}
\SE~produces an answer for every circuit $N$ and so it is complete.
Indeed, in every iteration of the while loop, \SE~excludes at least
one new input of $N$ that has not been excluded so far. So the set of
allowed inputs monotonically decreases.  Eventually, \SE~either
excludes all inputs but \pnt{\hat{x}} or derives a clause
$\overline{z}$. In either case, \SE~terminates returning a result.

\SE~is sound. It reports that $N$ is correct in two cases. First, when
a clause $\overline{z}$ is generated and \SE~proves that
\Impl{F}{\overline{z}}.  Then $F \wedge z$ is unsatisfiable and hence
no counterexample exists.  Second, \SE~reports that $N$ is correct
when all inputs of $N$ but \pnt{\hat{x}} are excluded by added
clauses. Since \SE~ guarantees that a new clause cannot remove all
counterexamples and $N$ evaluates to 0 at \pnt{\hat{x}}, this means
that $N$ evaluates to 0 for all inputs.  \SE~reports that $N$ is buggy
when an assignment satisfying $F \wedge z$ is found. This assignment
specifies an input for which $N$ evaluates to 1.

%
%
\subsection{Two modifications of the algorithm}
\label{subsec:modifications}

In this subsection, we consider two modifications of \SE~ that can be
used to improve its efficiency. The first modification is given in
Figure~\ref{fig:use_sat}. The high-lighted part shows the lines added
to the code of \SE~of Figure~\ref{fig:impl_sim_by_excl}.

\input{m1odif_fig}

The main idea here is to occasionally run a light SAT-check on formula
$F \wedge z$.  We assume that the effort of the SAT-solver is limited
by some parameter specifying e.g. the maximum number of backtracks.
Such a SAT-check is invoked when the number of new clauses added to
$F$ exceeds a threshold. Such a strategy makes sense because clauses
added to $F$ by \SE~are not derived by resolution. Deriving such
clauses by resolution may be dramatically more complex, which makes
them quite valuable for a resolution based SAT-solver.

Consider for instance, derivation of clause $C = x_1 \vee x_2$ for
circuit $N$ shown in Figure~\ref{fig:sim_by_excl} (see discussion of
Subsection~\ref{subsec:comp_two_sims}). Assume that circuit $N$ has
only one output specified by variable $z$. Suppose that formula $F
\wedge z$ is satisfiable (i.e. $N$ is buggy) and $C$ excludes at least
one counterexample. Then $C$ is not even implied by $F \wedge z$ and
thus cannot be derived by resolution. However, even if $C$ is implied
by $F \wedge z$, its resolution derivation may be very hard and
involve many clauses of $F$. On the other hand, a PQE-solver is
capable of concluding that $C$ can be safely added to $F$ just by
examining the clauses specifying gates $G_1$,$G_2$,$G_3$ and using the
fact that variables $x_1,x_2,y_1,y_2$ do not feed any other gates.

\input{m2odif_fig}

The second modification is shown in Figure~\ref{fig:use_sim}. In this
modification one occasionally runs regular simulation. We assume that
only tests that have not been excluded yet are generated. These are
the tests satisfying formula $G$.  A simulation run is performed when
the number of new clauses added to $G$ exceeds a threshold. We assume
that the number of tests generated in a simulation run is limited by
parameter \ti{\#tries}. Performing occasional simulation runs makes
sense because adding new clauses reduces the search space and hence
increases the probability of generating a counterexample.

%% file: s3at_fig.tex
%
%
\setlength{\intextsep}{2pt}
\setlength{\textfloatsep}{10pt}
\begin{figure}
\small
\vspace{-10pt}
\begin{tabbing}
aaa\=bb\=cc\= dd\= eeeeeeeeeeeeeeeeee\= \kill
// $N(X,Y,z)$ is a single-output circuit \\
// $X$ is the set of  input variables \\
// $Y$ is the set of  internal variables \\
// $z$ specifies the output variable \\
// $\mi{SimByExcl}$ returns a counterexample \\
// ~~~~~or \ti{nil} if no counterexample exists \\
//\\
$\mi{SimByExcl}(N)$\{\\
\tb{\scriptsize{1}}\> $F(X,Y,z) := \mi{GenCnfForm}(N)$; \\
\tb{\scriptsize{2}}\> $\pnt{\hat{x}} := \mi{GenInp}(X)$; \\
\tb{\scriptsize{3}}\> $\mi{val}(z) := \mi{Simulate}(N,\pnt{\hat{x}});$\\
\tb{\scriptsize{4}}\> if ($\mi{val}(z) = 1$) return(\pnt{\hat{x}}); \\
\tb{\scriptsize{5}}\> $\cof{C}{\hat{x}} := \mi{LongestFalsifClause}(\pnt{\hat{x}})$; \\
\tb{\scriptsize{6}}\> $G := \emptyset$; \\
$----------------$ \\
\tb{\scriptsize{7}}\> while (\ti{true}) \{\\
\tb{\scriptsize{8}}\Tt   $\pnt{x} := \mi{FindSatAssgn}(G \wedge \cof{C}{\hat{x}})$;\\
\tb{\scriptsize{9}}\Tt   if ($\pnt{x} = \mi{nil}$) return($\mi{nil}$); \\
\tb{\scriptsize{10}}\Tt   $C := \mi{GenFalsifClause}(\pnt{x},\pnt{\hat{x}})$; \\ 
\tb{\scriptsize{11}}\Tt   $C^* := \mi{SolvePQE}(\prob{W}{C \wedge F})$; // $W = X \cup Y$ \\
 ~~~~~~~~~~ \\
\tb{\scriptsize{12}}\Tt  if ($C^* = \overline{z}$) break;\\
\tb{\scriptsize{13}}\Tt  $F := F \wedge C$; \\
\tb{\scriptsize{14}}\Tt   $G:= G \wedge C$; \}\\
$----------------$ \\
\tb{\scriptsize{15}}\> if $((\overline{C} \wedge F \wedge z) \equiv \mi{false})$ return($\mi{nil}$); \\
\tb{\scriptsize{16}}\> $(\pnt{x},\pnt{y},z) := \mi{SatAssgn}(\overline{C} \wedge F \wedge z)$;\\
\tb{\scriptsize{17}}\> return(\pnt{x}); \}\\
\end{tabbing} 
\vspace{-10pt}
\caption{The basic algorithm of simulation-by-exclusion}
\label{fig:impl_sim_by_excl}
\end{figure}

%% file: m1odif_fig.tex
%
%
\setlength{\intextsep}{4pt}
\setlength{\textfloatsep}{10pt}
\begin{figure}
\small
\vspace{10pt}
\begin{tabbing}
aaa\=bb\=cc\= dd\= eeeeeeeeeeeeeeeeee\= \kill
$\mi{SimByExcl}(N)$\{\\
$~~~~~\ldots$ \\
\tb{\scriptsize{6}}\> $G := \emptyset$;  \\
\tb{\scriptsize{6.1}}\> \hl{$\mi{NumNewClauses : = 0}$;} \\
\tb{\scriptsize{7}}\> while (\ti{true}) \{\\
$~~~~~\ldots$ \\
\tb{\scriptsize{12}}\Tt  if ($C^* = \overline{z}$) break;  \\
\tb{\scriptsize{13}}\Tt    $F := F \wedge C$; \\
\tb{\scriptsize{14}}\Tt    $G:= G \wedge C$; \\
\tb{\scriptsize{14.1}}\Tt  \hl{$\mi{NumNewClauses} := \mi{NumNewClauses} + 1$;} \\
\tb{\scriptsize{14.2}}\Tt   \hl{if $(\mi{ NumNewClauses} > \#threshold)$ \{} \\
\tb{\scriptsize{14.3}}\ttt  \hl{$\mi{NumNewClauses} :=\emptyset$;} \\
\tb{\scriptsize{14.4}}\ttt    \hl{$\mi{answer} := \mi{SatCheck}(F \wedge z,\#\mi{backtracks})$;} \\
\tb{\scriptsize{14.5}}\ttt    \hl{if ($\mi{answer} \neq \mi{unfinished}$) return($\mi{answer}$);\}\}} \\
$~~~~~\ldots$ \\
\end{tabbing} 
\vspace{-10pt}
\caption{Modification of the algorithm shown in Figure~\ref{fig:impl_sim_by_excl}}
\label{fig:use_sat}
\end{figure}

%% file: m2odif_fig.tex
%
%

\setlength{\intextsep}{4pt}
\setlength{\textfloatsep}{10pt}
\begin{figure}
\small
\vspace{10pt}
\begin{tabbing}
aaa\=bb\=cc\= dd\= eeeeeeeeeeeeeeeeee\= \kill
$\mi{SimByExcl}(N)$\{\\
$~~~~~\ldots$ \\
\tb{\scriptsize{6}}\> $G := \emptyset$;  \\
\tb{\scriptsize{6.1}}\> \hl{$\mi{NumNewClauses : = 0}$;} \\
\tb{\scriptsize{7}}\> while (\ti{true}) \{\\
$~~~~~\ldots$ \\
\tb{\scriptsize{12}}\Tt  if ($C^* = \overline{z}$) break;  \\
\tb{\scriptsize{13}}\Tt    $F := F \wedge C$; \\
\tb{\scriptsize{14}}\Tt    $G:= G \wedge C$; \\
\tb{\scriptsize{14.1}}\Tt  \hl{$\mi{NumNewClauses} := \mi{NumNewClauses} + 1$;} \\
\tb{\scriptsize{14.2}}\Tt   \hl{if $(\mi{ NumNewClauses} > \#threshold)$ \{} \\
\tb{\scriptsize{14.3}}\ttt  \hl{$\mi{NumNewClauses} :=\emptyset$;} \\
\tb{\scriptsize{14.4}}\ttt     \hl{$\mi{Cex} := \mi{RunSim}(N,G,\#tries)$;} \\
\tb{\scriptsize{14.5}}\ttt    \hl{if ($\mi{Cex} \neq \mi{nil}$) return($\mi{Cex}$);\}\}} \\
$~~~~~\ldots$ \\
\end{tabbing} 
\vspace{-10pt}
\caption{Another modification of the algorithm of Figure~\ref{fig:impl_sim_by_excl}}
\label{fig:use_sim}
\end{figure}

%% file: background.tex
\section{Background}
\label{sec:background}
The existing methods for checking equivalence of combinational
circuits $N'$ and $N''$ can be roughly partitioned into two groups.
The first group consists of methods that do not try to exploit the
structural similarity of $N'$ and $N''$. An example of an equivalence
checker of this group is an algorithm that builds separate
BDDs~\cite{bryant_bdds1} of $N'$ and $N''$ and then checks that these
BDDs are identical. One more example, is an equivalence checker
constructing a CNF formula that is satisfiable only if $N'$ and $N$''
are not equivalent. (In Subsection~\ref{subsec:method_comp} we
explained how such a formula is generated.) This equivalence checker
runs a SAT-solver~\cite{grasp,chaff} to test the satisfiability of
this formula.  Unfortunately, the algorithms of this group do not
scale well with the size of $N',N''$.

Methods of the other group try to take into account the similarity of
$N',N''$~\cite{kuehlmann97,date01}. These methods work well even for
quite large circuits $N',N''$ if they are very similar structurally. A
flaw of such methods is that they do not scale well as $N',N''$ become
more dissimilar. As we mentioned earlier, the method based on CRR that
we introduced in this paper can be arguably used to design an
equivalence checker that take the best of both worlds.  That is it
scales better than the methods of the first group in terms of circuit
size and improves on the scalability of the methods of the second
group in terms of circuit dissimilarity.

Simulation is still the main workhorse of verification. So a lot of
research has been done to improve its efficiency and effectiveness.
For example, symbolic simulation~\cite{SymbolSim} is used to run many
tests at once.  Constrained random simulation generates tests aimed at
achieving a particular goal e.g. better coverage with respect to a
metric~\cite{constr_rand}.  However, all these methods are based on
the paradigm of decreasing the complexity of range computation by
reducing the number of inputs to be considered at once. To the best of
our knowledge none of existing methods exploits the idea of computing
range reduction introduced in~\cite{tech_rep_crr}. There we explored
this idea in the context of model checking. In the current paper, we
contrast simulation based on computing range reduction with
traditional simulation.

%% file: c3onclusions.tex
\section{Conclusions}
\label{sec:conclusions}
We described new methods of equivalence checking and simulation based
on Computing Range Reduction (CRR). Our interest in these methods is
twofold. First, they allow one to take into account subtle structural
properties of the design. For instance, equivalence checking based on
CRR allows can employ a ``cut advancement'' strategy. A limited
version of this strategy has been very successful in equivalence
checking of industrial designs. Second, one can argue that CRR can be
efficiently performed using a technique called partial quantifier
elimination that we are working on.

%% file: main.bbl
\begin{thebibliography}{10}

\bibitem{bryant_bdds1}
R.~Bryant.
\newblock Graph-based algorithms for {B}oolean function manipulation.
\newblock {\em IEEE Transactions on Computers}, C-35(8):677--691, August 1986.

\bibitem{SymbolSim}
R.~Bryant.
\newblock Symbolic simulation---techniques and applications.
\newblock In {\em DAC-90}, pages 517--521, 1990.

\bibitem{fmcad12}
E.Goldberg and P.Manolios.
\newblock Quantifier elimination by dependency sequents.
\newblock In {\em FMCAD-12}, pages 34--44, 2012.

\bibitem{fmcad13}
E.Goldberg and P.Manolios.
\newblock Quantifier elimination via clause redundancy.
\newblock In {\em FMCAD-13}, pages 85--92, 2013.

\bibitem{tech_rep_crr}
E.Goldberg and P.Manolios.
\newblock Bug hunting by computing range reduction.
\newblock Technical Report arXiv:1408.7039 [cs.LO], 2014.

\bibitem{hvc-14}
E.~Goldberg and P.~Manolios.
\newblock Partial quantifier elimination.
\newblock In {\em Proc. of HVC-14}, pages 148--164. Springer-Verlag, 2014.

\bibitem{date01}
E.~Goldberg, M.~Prasad, and R.~Brayton.
\newblock Using {SAT} for combinational equivalence checking.
\newblock In {\em {DATE}}, pages 114--121, 2001.

\bibitem{constr_rand}
O.~Guzey and L.C. Wang.
\newblock Coverage-directed test generation through automatic constraint
  extraction.
\newblock In {\em HLDVT}.

\bibitem{kuehlmann97}
A.~Kuehlmann and F.~Krohm.
\newblock {Equivalence Checking Using Cuts And Heaps}.
\newblock {\em DAC}, pages 263--268, 1997.

\bibitem{grasp}
J.~Marques-Silva and K.~Sakallah.
\newblock Grasp -- a new search algorithm for satisfiability.
\newblock In {\em ICCAD-96}, pages 220--227, 1996.

\bibitem{chaff}
M.~Moskewicz, C.~Madigan, Y.~Zhao, L.~Zhang, and S.~Malik.
\newblock Chaff: engineering an efficient sat solver.
\newblock In {\em DAC-01}, pages 530--535, New York, NY, USA, 2001.

\end{thebibliography}
